\documentclass[12pt]{amsart}
\usepackage{listings} 
\usepackage{comment}
\usepackage{hyperref}
\usepackage{amssymb, amsmath,subfigure}
\usepackage{amsfonts}
\usepackage{float}
\usepackage{graphics,color}
\usepackage{graphicx}
\graphicspath{{./figures/}}

\newcommand{\mat}[2][ccccccccccccccccccccccc]{\left[\begin{array}{#1}
                                   #2\\
                                   \end{array}
                                   \right]}

\def\Pr{\text{Pr}}

 \newtheorem{theorem}{Theorem}[section]
\newtheorem{lemma}[theorem]{Lemma}

\newtheorem{remark}[theorem]{Remark}

\begin{document}
\title[2D Rayleigh--B{\'e}nard Convection]{Pattern Formations of 2D Rayleigh--B{\'e}nard Convection with No-Slip Boundary Conditions for the Velocity at the Critical Length Scales}
\author{Taylan Sengul, Jie Shen, Shouhong Wang}
\begin{abstract}
We study the Rayleigh-B{\'e}nard convection in a 2-D rectangular domain with no-slip boundary conditions for the velocity. The main mathematical challenge  is due to the no-slip boundary conditions, since the separation of variables for the linear eigenvalue problem which works in the free-slip case is no longer possible. 
It is well known that as the Rayleigh number crosses a critical threshold $R_c$,  the system bifurcates to  an attractor, which  is an $(m-1)$--dimensional sphere, where $m$ is the number of eigenvalues which cross zero as R crosses $R_c$. 
The main objective of this article is to derive a full classification of the  structure of this bifurcated attractor when $m=2$.  More precisely, we rigorously prove that when $m=2$, the bifurcated attractor is homeomorphic to a one-dimensional circle consisting of exactly  four or eight steady states and their connecting heteroclinic orbits. In addition, we show that  the mixed modes can be stable steady states for small Prandtl numbers.
\end{abstract}
\maketitle
\section{Introduction}
The Rayleigh-B\'enard convection problem is one of the fundamental problems in the physics of fluids. The basic phenomena of the Rayleigh-B{\'e}nard convection in horizontally extended systems are widely known. The influence of  the side walls, although not studied as throughly as the horizontally extended case, is of practical importance for engineering applications. 

In this paper we study the Rayleigh-B{\'e}nard convection in a 2-D rectangular domain with no-slip boundary conditions for the velocity. This problem is also closely related to the problem of infinite channel with rectangular cross-section which has been studied by  Davies-Jones \cite{davies70}, Luijkx--Platten \cite{luijkx1981} and Kato--Fujimura \cite{kato2000} among others.

The linear aspects of the problem we consider in this paper have been studied by Lee--Schultz--Boyd \cite{lee1989}, Mizushima \cite{mizushima95} and Gelfgat \cite{gelfgat99}. In these papers, the critical Rayleigh number and the structure of the critical eigenmodes have been studied for small aspect ratio containers.

From dynamical transition and pattern formation point of view,  Ma and Wang \cite{ma2004dynamic, ma2007rayleigh} proved that under some general boundary conditions, the problem always undergoes a dynamic transition  to an attractor $\Sigma_R$ as the Rayleigh number $R$ crosses the first critical Rayleigh number $R_c$. They also proved that the bifurcated attractor, homological to $S^{m-1}$,  where $m$ is the number of critical eigenmodes.

In the 2-D setting that we consider, $m$ is either 1 or 2 and the latter case can only happen at the critical length scales where two modes with wave numbers $k$ and $k+1$ become critical simultaneously.  When $m=1$, the structure of $\Sigma_R$ is trivial which is merely a disjoint union of two attracting steady states. Thus our task in this paper is to classify the structure of the attractor when $m=2$. This has been studied recently in \cite{sengulRB} for the 3D Rayleigh-B{\'e}nard problem where the boundaries were assumed to be free-slip for the velocity and the wave numbers of the critical modes were assumed to be equal.

The main mathematical challenge in this paper is due to the no-slip boundary conditions  since the separation of variables for the linear eigenvalue problem which works in the free-slip case is not possible anymore. To overcome this difficulty, the main approach for our study is to combine rigorous analysis and numerical computation using spectral method.  

As we know, spectral methods have long been used to address the hydrodynamic instability problems. In fact, in his seminal work \cite{orszag71}, Orszag studied the classical Orr-Sommerfeld linear instability problem using a Chebyshev-tau method. In this paper, to treat the linear eigenvalue problem, we employ a Legendre-Galerkin method where compact combinations of Legendre polynomials, called generalized Jacobi polynomials, satisfying all the boundary conditions are used as trial functions. The main advantage of our Legendre-Galerkin method  is that  the resulting matrices are sparse which allows a very efficient and accurate solution of the linearized problem; see also Hill--Straughan \cite{hill2006legendre} and Gheorghiu--Dragomirescu \cite{gheorghiu2009spectral}.

Once the eigenpairs of the linear problem are identified, the transition analysis is carried out by reducing the infinite dimensional system to the center manifold in the first two critical eigendirections. The coefficients of this reduced system are calculated numerically. Our main results are described below. 

We first classify the eigenmodes into four classes according to their parities using the symmetry of the problem. Then we numerically show that the first two unstable modes are always parity class one or two. Then we study the transition near the critical length scales where two eigenvalues become positive simultaneously. Next, we rigorously prove that the local attractor at small supercritical Rayleigh numbers is in fact homeomorphic to the circle which has four or eight steady states with half of them as stable points and the rest as saddle points. The critical eigenmodes are always bifurcated steady states on the attractor and when the attractor has eight steady states, the mixed modes which are superpositions of the critical eigenmodes are also bifurcated.

Second, let $\beta_1$ and $\beta_2$ denote the two largest eigenvalues of the linearized problem. We find that a small neighborhood of $\beta_1$ = $\beta_2=0$ in the $\beta_1$--$\beta_2$ plane can be separated into several sectors with different asymptotical structures. In particular, we find that there is a critical Prandtl  number $\Pr_c$ for the first two critical length scales $L=1.5702$ and $L=2.6611$, such that for $\Pr<\Pr_c$, there is a sector in this plane for which mixed modes are stable fixed points of the attractor. For $\Pr>\Pr_c$, the mixed modes are never stable and instead there is a sector in this plane in which both of the critical eigenmodes coexist as stable steady states. In this case, the initial conditions determine which one of these eigenmodes will be realized. The critical Prandtl number is around $0.14$ for the first critical length scale $L=1.5702$ and around $0.05$  for the second critical length scale $L=2.6611$. For higher critical length scales we found that mixed modes are never stable points of the attractor.

Third, recently Ma--Wang has developed the dynamic transition theory to study transition and bifurcation problems in nonlinear sciences; see \cite{ptd}. This paper is a first attempt to combine this theory with the numerical tools of the spectral methods to study the detailed structure  of the transition and pattern formation.

The paper is organized as follows. In Section 2, the governing equations and the functional setting of the problem is discussed. In Section 3, linear eigenvalue problem is studied. Section 4 states the main theorem. Section 5 is devoted to the proof of the main theorem. In Section 6, we demonstrate a method to compute the coefficients of the reduced system. And the last section discusses the results obtained by our analysis. 
\section{Governing Equations and the Functional Setting}
Two dimensional thermal convection with no-slip, perfectly conducting boundaries can be modeled by the Boussinesq equations. The governing equations on the rectangular domain $\Omega = (0,L)\times(0,1) \in \mathbb{R}^2$ read as
\begin{equation} \label{rbm1}
\begin{aligned}
& \frac{\partial \mathbf{u}}{\partial t}+(\mathbf{u}\cdot \nabla ) \mathbf{u} 
=  -\Pr(\nabla p+\Delta \mathbf{u})+\sqrt{R} \sqrt{\Pr}\, \theta \, {\bf k} , \\ 
& \frac{\partial \theta }{\partial t}+(\mathbf{u}\cdot \nabla) \theta  = \sqrt{R} \sqrt{\Pr} w+\Delta \theta, \\ 
& \nabla \cdot\mathbf{u} =0.
\end{aligned}
\end{equation}
Here $\mathbf{u} = (u,w)$ is the velocity field, $\theta$ is the temperature field and $p$ is the pressure field. These fields represent a perturbation around the motionless state with a linear temperature profile. The dimensionless numbers are the Prandtl number Pr and the Rayleigh number R which is also the control parameter. $\mathbf{k}$ represents the unit vector in the z-direction.

The equations \eqref{rbm1} are supplemented with no-slip boundary conditions for the velocity and perfectly conducting boundary conditions for the temperature.
\begin{equation} \label{bc}
\mathbf{u} = \theta = 0, \qquad \text{on $\partial \Omega$}.
\end{equation}

For the functional setting, we define the relevant function spaces:
\begin{equation}\label{func spaces}
\begin{aligned} 
& H =\left\{ (\mathbf{u},\theta)\in L^{2}\left(
\Omega,\mathbb{R}^3 \right) :\nabla\cdot\mathbf{u}=0,\mathbf{u}\cdot n\mid_{\partial \Omega}=0\right\},\\
& H_1 =\left\{ (\mathbf{u},\theta)\in H^{2}\left(
\Omega, \mathbb{R}^3 \right) :\nabla\cdot\mathbf{u}=0,\mathbf{u} \mid_{\partial \Omega}= 0, \, \theta \mid_{\partial \Omega}=0\right\}.
\end{aligned}
\end{equation}
For $\phi=({\bf u},\theta)$, let $G:H_1\rightarrow H$ and $L_R:H_1\rightarrow H$ be defined by
\begin{equation}\label{operators}
\begin{aligned}
& L_R\phi=\left(\mathcal{P}(\Pr \Delta {\bf u}+\sqrt{R}\sqrt{\Pr}\, \theta {\bf k}),\, \sqrt{R} \sqrt{\Pr} \,w+\Delta\theta \right),\\
& G(\phi)=-\left(\mathcal{P}({\bf u}\cdot\nabla){\bf u} ,\,({\bf u}\cdot\nabla)\theta)\right),
\end{aligned}
\end{equation}
with $\mathcal{P}$ denoting the Leray projection onto the divergence-free vectors.

The equations \eqref{rbm1}--\eqref{bc} supplemented with initial conditions can be put into the following abstract ordinary differential equation:
\begin{equation} \label{functional}
\frac{d\phi}{dt} = L_R \phi +G(\phi), \qquad \phi(0)=\phi_0.
\end{equation}

For results concerning the existence and uniqueness of solutions of \eqref{functional}, we refer  to Foias, Manley, and Temam \cite{foias87}. 

Finally for $\phi_i=({\bf u}_i,\theta_i)$, ${\bf u}_i = (u_i,w_i)$, $i=1,2,3$ we define the following trilinear forms.
\begin{equation} \label{bilinear}
\begin{aligned}
& G(\phi_{1},\phi_{2},\phi_{3})=-\int_{\Omega}({\bf u}_1 \cdot\nabla){\bf u}_2\cdot{\bf u}_3 dx dz-\int_{\Omega}({\bf u}_1\cdot\nabla)\theta_{2}\cdot \theta_{3} dx dz,\\
& G_s(\phi_1,\phi_2,\phi_3)= G(\phi_1,\phi_2,\phi_3)+ G(\phi_2,\phi_1,\phi_3).
 \end{aligned}
\end{equation}

\section{Linear Analysis}
We first study the eigenvalue problem $L_R \phi = \beta \phi$ which reads as
\begin{equation} \label{lineig}
\begin{aligned}
& \Pr ( \Delta u -\frac{\partial p}{\partial x}) = \beta u, \\
& \Pr ( \Delta w -\frac{\partial p}{\partial z}) + \sqrt{R}\sqrt{\Pr}\, \theta = \beta w, \\
& \Delta \theta + \sqrt{R} \sqrt{\Pr}\,w = \beta \theta,\\
& \text{div} {\bf u} = 0,\\
& u = \theta = 0, \quad \text{at } \partial \Omega.
\end{aligned}
\end{equation}
Below we list some of the properties of this eigenvalue problem.
\begin{itemize}
\item[1)] The linear operator $L_R$ is symmetric. Hence the eigenvalues $\beta_i$ are real and the eigenfunctions $\phi_i$ are orthogonal with respect to $L^2$--inner product. Moreover there is a sequence 
\[
0<R_1\leq R_2 \leq \cdots
\] 
such that $\beta_i(R_i) = 0$. $R_i$ is found by setting $\beta =0$ in \eqref{lineig}. In this case the problem becomes an eigenvalue problem with $\sqrt{R}$ as the eigenvalue. 
\item[2)] We have
\[
\beta_i(R)\overset{>}{\underset{<}{=}}0 \, \text{if } R\overset{>}{\underset{<}{=}}R_i,
\]
which can be seen by computing the derivative of $\beta_i$ with respect to $R$ at $R=R_i$. 
\begin{equation} \label{dbeta}
\frac{d\beta_i}{dR} \mid_{R=R_i}  =  \frac{1}{\sqrt{R_i}}\frac{ \sqrt{\Pr}\int_{\Omega} \theta_i w_i}{\int_{\Omega} u_i^2+ w_i^2 + \theta_i^2},
\end{equation}
where $(u_i,w_i,\theta_i)$ is the ith eigenfunction. Also at $R=R_i$, by the third equation in \eqref{lineig}$, w_i = -R_i^{-1/2} \Pr^{-1/2} \Delta \theta_i$ as $\beta_i(R_i)=0$. Plugging these into \eqref{dbeta} and integrating by parts, we see that
\[
\frac{d\beta_i}{dR} \mid_{R=R_i}  =  \frac{\int_{\Omega}  |\nabla\theta_i|^2  dx}{R_i\int_{\Omega} {(|u_i|^2  + |w_i|^2 + |\theta_i|^2 )dx}}>0
\]

\item[3)] We denote the critical Rayleigh number $R_c = R_1$. That is
\begin{equation} \label{PES}
\begin{aligned}
&  \beta_i(R)
\begin{cases}  
<0& \text{if $R<R_c $,}\\
=0& \text{if $R=R_c $,}\\
>0& \text{if $R>R_c $.}\\
\end{cases}
&& i=1,\dots,m\\
& \beta_i(R_c)<0, && i>m.
\end{aligned}
\end{equation}
$m$ in \eqref{PES} does not depend on the Prandtl number Pr but only on $L$. To see this, one makes the change of variable $\theta = \sqrt{\Pr} \, \theta'$ so that the solution of \eqref{lineig} for the eigenvalue of $\beta =0$ is independent of Pr. By simplicity of the first eigenvalue (see Theorem 3.7 in Ma--Wang \cite{bifbook}), for almost every value of $L$ except a discrete set of values, $m$ in \eqref{PES} is 1.
\end{itemize}

Introducing the streamfunction $\psi_z = u$, $\psi_x = -w$, we can eliminate the pressure $p$ from  the linear eigenvalue problem \eqref{lineig}.
\begin{equation}\label{rb2d01}
\begin{aligned} 
& \Pr \Delta^2\psi-\sqrt{R} \sqrt{\Pr}\theta_x = \beta(R) \, \Delta \psi,\\
-& \sqrt{R} \sqrt{\Pr} \psi_x+\Delta\theta= \beta(R) \theta,\\
& \psi=\frac{\partial \psi}{\partial n}=\theta=0 \text{ on }\partial \Omega.
\end{aligned}
\end{equation}

The linear equations \eqref{lineig} satisfy several discrete symmetries which may be found from the known non-trivial groups of continuous Lie symmetries of the field equations \eqref{rbm1}; see (Hydon\cite{hydon2000symmetry}, Marques--Lopez--Blackburn \cite{marques2004}). However, for the problem we consider, it can be easily verified that the linear equations have reflection symmetries about the horizontal and vertical mid-planes of the domain. Thus we can classify the solutions of the linear problem into four classes with different parities which are as defined in Table~\ref{tab:parity} where, for example, $\psi(o,e)$ means that $\psi$ is odd in the $x$-direction and even in the $z$-direction.
\renewcommand{\arraystretch}{1.5} 
\begin{table}
\centering
\begin{tabular}{ | c | c | c | c |}
\hline
Class 1 & Class 2 & Class 3 & Class 4\\
\hline 
 $\psi(e,e)$, $\theta(o,e)$ & $\psi(o,e)$, $\theta(e,e)$ & $\psi(e,o)$, $\theta(o,o)$ &  $\psi(o,o)$, $\theta(e,o)$ \\ \hline
\end{tabular}
\caption{\label{tab:parity} Possible parity classes of the eigenfunctions of the linear operator.}
\end{table}

We will employ a Legendre-Galerkin method (cf. Shen \cite{shen1994efficient}, Shen--Tang--Wang \cite{STW2011}) to solve the linear eigenvalue problem \eqref{rb2d01}. For this, first we transform the domain with the change of variable
\[
(x,z) \in  (0,L)\times(0,1) \rightarrow (X,Z) = (\frac{2x}{L}-1,2z-1) \in (-1,1)^2.
\]
The approximate solutions $(\psi^N,\theta^N)$ of \eqref{rb2d01} will be sought in the finite dimensional space 
\[
X^N = \text{span} \{ (e_j(x) e_k(z),f_l(x)f_m(z)) \mid  j,k,l,m=0,1,\dots,N-1\} , 
\]
where $e_j$ and $f_j$ are generalized Jacobi polynomials (cf. Guo--Shen--Wang\cite{GSW06}, Shen--Tang--Wang \cite{STW2011}) which satisfy the boundary conditions
\[
e_i(\pm 1)=De_i(\pm 1) = f_i(\pm 1) =0.
\]
Here $D$ denotes the derivative.
The polynomials $e_i$ and $f_i$ are defined as in Chapter 6 of Shen--Tang--Wang \cite{STW2011},
\begin{equation} \label{J11}
f_i(z) = L_i(z) - L_{i+2}(z),
\end{equation}
\begin{equation} \label{J22}
e_i(z)= \frac{ L_i(z)-\frac{4i+10}{2i+7}L_{i+2}(z)+\frac{2i+3}{2i+7}L_{i+4}(z) }{(2i+3)(4i+10)^{1/2}},
\end{equation}
where $L_i$ is the $i^{\text{th}}$ Legendre polynomial. The coefficient of $e_i$ guarantees that $(D^2e_i,D^2e_j) = \delta_{ij}$.

We write the approximate solutions of the equation \eqref{rb2d01} with coefficients to be determined by
\begin{equation}\label{exp1}
\psi^N = \sum_{j=0}^{N_x-1} \sum_{k=0}^{N_z-1} \tilde{\psi}^N_{jk} e_j(x) e_k (z), \quad \theta^N = \sum_{j=0}^{N_x-1} \sum_{k=0}^{N_z-1} \tilde{\theta}^N_{jk} f_j(x) f_k(z).
\end{equation}
Here $N=2N_x N_z$ denotes the total degrees of freedom.

Let us define for $i,j =0\dots,m-1$,
\begin{equation*}
\begin{aligned}
& (A_1^m)_{ij} = (D^2 e_i,D^2 e_j) = \delta_{ij}, && (A_2^m)_{ij} = (D^2 e_i, e_j)= -(De_j,De_i), \\
& (A_3^m)_{ij} = (e_i,e_j), && (A_4^m)_{ij} = (e_i,f_j), \\
& (A_5^m)_{ij} = (D^2 f_i,f_j), && (A_6^m)_{ij} = (f_i,f_j), \\
& (A_7^m)_{ij} = (Df_i,e_j), 
\end{aligned}
\end{equation*}
and for $j=0,\dots,N_x-1$, $k=0,\dots,N_z-1$,
\[
\tilde{\psi}^N=\{\tilde{\psi}^N_{jk}\}, \tilde{\theta}^N=\{\tilde{\theta}^N_{jk}\}.
\]
Using the following property of the Legendre polynomials
\begin{equation} \label{legformula}
(2i+3) L_{i+1} = D(L_{i+2}-L_i),
\end{equation}
it is easy to see that:
\begin{equation}\label{computeAij}
D e_i(z) = \frac{L_{i+3}-L_{i+1}}{\sqrt{4i+10}}, \,  D^2e_i(z)=\sqrt{\frac{2i+5}{2}} L_{i+2}(z), \, Df_i = -(2i+3) L_{i+1}.
\end{equation}
By \eqref{computeAij}, it is easy to determine the elements of the matrices $A_k$ by the orthogonality of the Legendre polynomials. In particular, the matrices $A_1^m,\dots A_7^m$ are banded, and except for $A_4^m$ and $A_7^m$, they are symmetric. 

Putting \eqref{exp1} into \eqref{rb2d01}, multiplying the resulting equations by $e_m(x)e_n(z)$ and $f_m(x)f_n(z)$ respectively and integrating over $-1\leq x \leq 1$, $-1\leq z \leq1$ to obtain
\begin{equation} \label{rb2d4}
B^N \bar{x}^N - \sqrt{R} C^N \bar{x}^N = \beta^N(R) D^N\bar{x}^N.
\end{equation}

Here: 
\begin{equation}\label{rb2d5}
\begin{aligned} 
& B^N = \mat{\Pr \,X_1 & 0 \\ 0 & X_3}_{N \times N}, && C^N = \mat{0 & \sqrt{\Pr}\,X_2 \\ -\sqrt{\Pr}\,X_2^T & 0}_{N \times N}, \\
& D^N = \mat{X_4 & 0 \\ 0 & X_5}_{N \times N}, && \bar{x}^N = \mat{\text{vec}(\tilde{\psi}^N) \\ \text{vec}(\tilde{\theta}^N)}_{N\times 1},
\end{aligned}
\end{equation}
where
\begin{equation}\label{rb2d6}
\begin{aligned} 
& X_1 = \frac{2^4}{L^4} A_3^{N_z} \otimes A_1^{N_x} + \frac{2^5}{L^2} A_2^{N_z} \otimes A_2^{N_x} + 2^4 A_1^{N_z} \otimes A_3^{N_x}, \\
& X_2 =   \frac{2}{L}  A_4^{N_z} \otimes (A_7^{N_x})^T, \quad X_3 = \frac{2^2}{L^2} A_6^{N_z} \otimes A_5^{N_x} + 2^2 A_5^{N_z} \otimes A_6^{N_x},\\
& X_4 = \frac{2^2}{L^2} A_3^{N_z} \otimes A_2^{N_x} + 2^2 A_2^{N_z} \otimes A_3^{N_x}, \quad X_5 = A_6^{N_z} \otimes A_6^{N_x}.
\end{aligned}
\end{equation}
In \eqref{rb2d5} and \eqref{rb2d6} we use the following notations. For  a $m\times k$ matrix $M$, $\text{vec}(M)$ is the $mk\times 1$ column vector obtained by concatenating the columns $M_i$ of $M$, i.e.
\[
\text{vec}(\mat{M_1 & M_2 & \cdots & M_k}) = \mat{M_1 & M_2 & \cdots & M_k}^T.
\]
$0$ stands for the zero matrix and $A \otimes B = \{ a_{ij} B \}_{i,j = 0,1,\dots,q}$ is the Kronecker product of $A$ and $B$. To obtain \eqref{rb2d4}, we used the following properties of the Kronecker product.
\[
\text{vec}(A X B) = (B^T \otimes A)\text{vec}(X), \quad (A \otimes B)^T = A^T \otimes B^T.
\]

We note that the matrices $B^N$, $C^N$ and $D^N$ in \eqref{rb2d4} are sparse, $B^N$ and $D^N$ are symmetric while $C^N$ is skew-symmetric.

From our linear analysis, we find the following results.
\begin{itemize}
\item Our numerical analysis suggest that $N_x = 6+2k \approx 6+ 2L $  and $N_z =8$ is enough to resolve the critical Rayleigh number and the first critical mode which has $k$ rolls in its stream function. We have checked that increasing $N_x$ and $N_z$ by two only modifies the fourth or fifth significant digit of the result.
\item In Figure~\ref{fig:first5evecs}, the first critical mode is shown for the length scales $L=1,\dots,5$. Note that the first critical stream function and the temperature distribution has always even parity in the z-direction while their x-parity alternates between odd and even as the length scale increases. As observed in Mizushima \cite{mizushima95}, we also verify the existence of the Moffatt vortices on the corners of the domain which are due to corner singularities as shown in Figure~\ref{fig:moffattvortex}.
\begin{figure}
\centering
\includegraphics[scale=0.4]{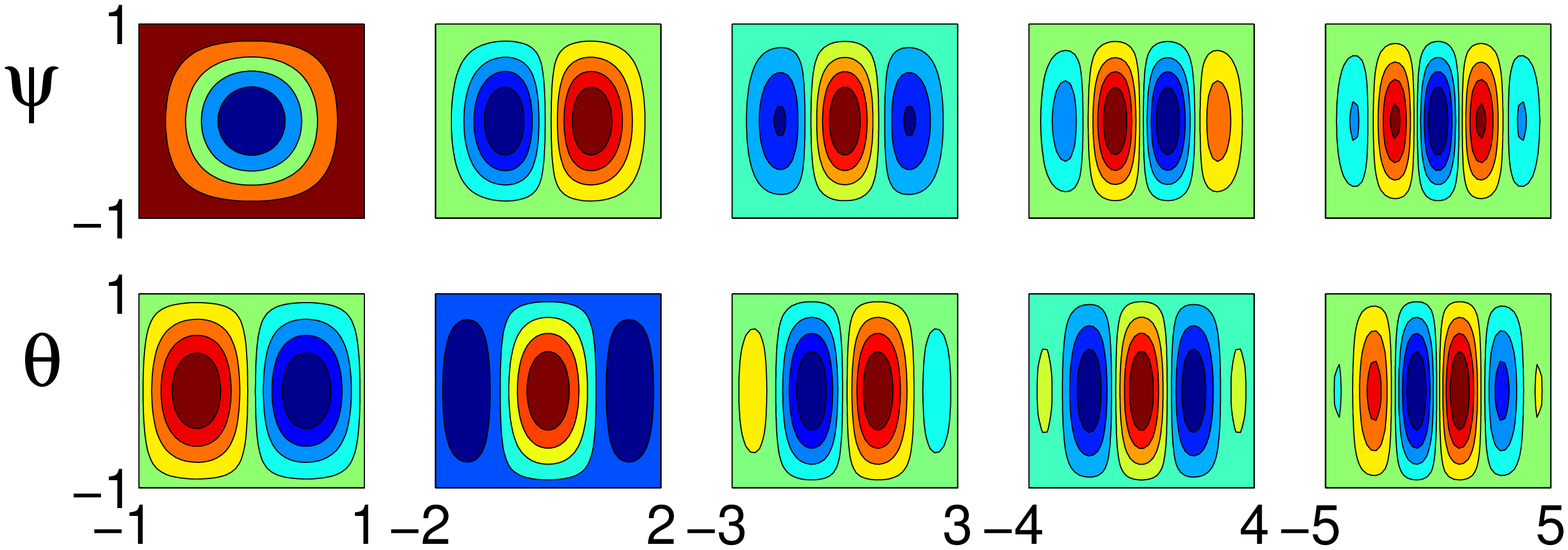}
\caption{ \label{fig:first5evecs}} $\psi$ (on top) and $\theta$ (on bottom) of the first critical mode on $(-L,L)\times(-1,1)$ for $L=1,\dots,5$.
\end{figure}
\begin{figure}
\centering
\includegraphics[scale=.3]{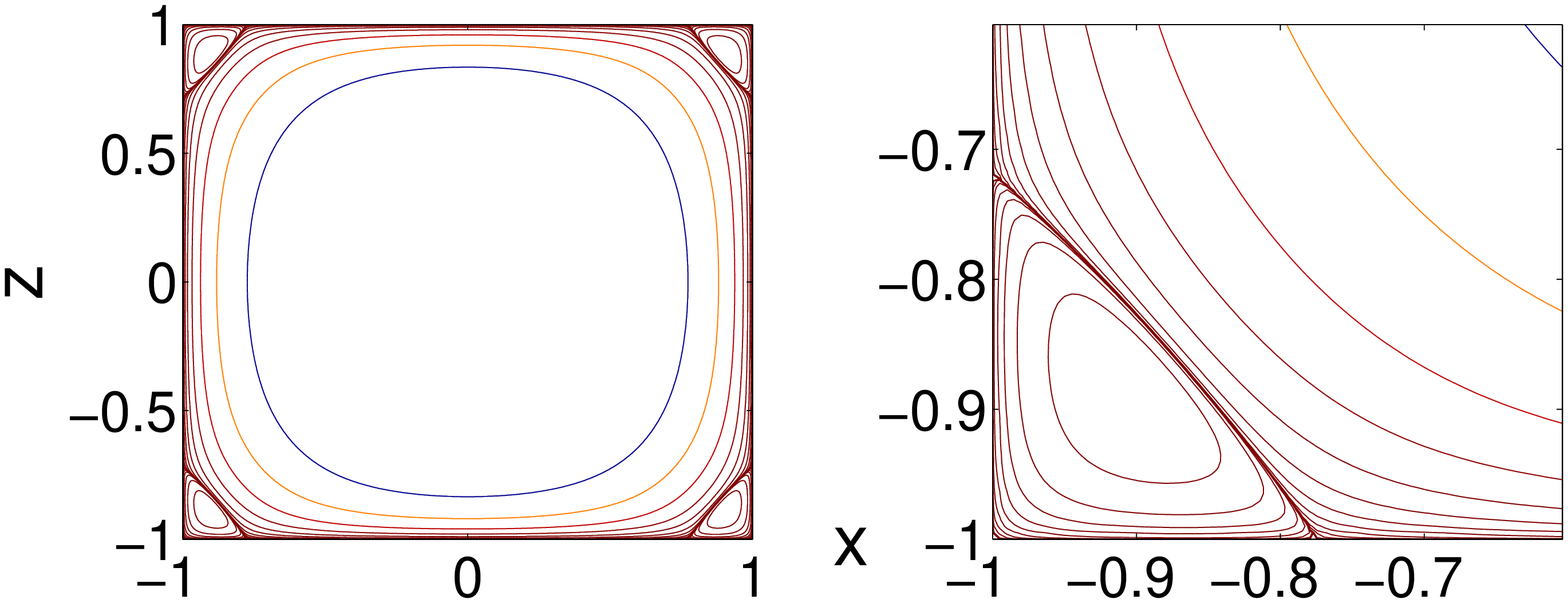}
\caption{ \label{fig:moffattvortex}} The left figure shows the plot of the first critical stream function for $L=1$ (the top left plot in Figure~\ref{fig:first5evecs}). The right figure shows the enlarged plot at the corner.
\end{figure}
\item For $L<21$ we observed that $m$ in \eqref{PES} is either 1 or 2. Moreover, $m=2$ only at the critical length scales which are  given in Table~\ref{asymp}. The results found are in agreement with those  in Mizushima \cite{mizushima95} and Lee-Schultz-Boyd \cite{lee1989}.  
\begin{table}
\begin{center}
\begin{tabular}{ |l | l|l|l||l|l|l|l|}
\hline
\textbf{L$_c$}&\textbf{k}&\textbf{R$_c$}&\textbf{N$_x$}&\textbf{L$_c$}&\textbf{k}&\textbf{R$_c$}&\textbf{N$_x$}\\\hline
1.5702&1&3086.6554&8&6.7711&6&1764.3754&18\\\hline
2.6611&2&2113.776&10&8.7992&8&1740.9174&22\\\hline
3.7048&3&1906.3395&12&10.8229&10&1729.5398&26\\\hline
4.7329&4&1826.4099&14&15.8738&15&1717.805&36\\\hline
5.7541&5&1786.8833&16&20.9197&20&1713.5226&46\\\hline
\end{tabular}
\caption{\label{asymp} At $L = L_c$, two modes become unstable. One of the modes has $k$ and the other one has $k+1$ rolls in x-direction in their stream functions. The critical Rayleigh number at this length scale is $R_c$. $N_x$ and $N_z=8$ are the number of polynomials used in the $x$ and $z$ directions respectively. }
\end{center}
\end{table}
\item The marginal stability curves of the first few critical eigenvalues are given in Figure~\ref{fig:par34muchbigger}. The figure demonstrates that the parities of the first two critical modes can only be of parity class one or two as given by Table~\ref{tab:parity}.

\begin{figure}
\centering
\includegraphics[scale=.3]{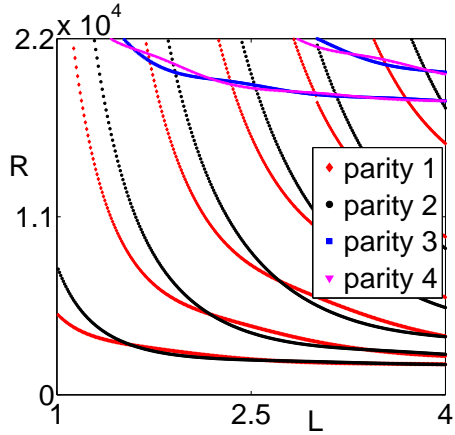}
\caption{ The marginal stability curves of the first few eigenvalues with eigenfunctions of different parity classes.  \label{fig:par34muchbigger}}
\end{figure}

Also it is seen in these figures that there is a repulsion of the eigenvalues. Namely the neutral stability curves of the same parity type do not intersect each other. Such a repulsion does not occur for free-slip boundary conditions. This repulsion arises from a structural instability of the transform of matrices into a Jordan canonical form and a detailed analysis can be found in Mizushima--Nakamura \cite{mizushima2002repulsion}.

\end{itemize}

\section{Main Theorem}
Let $m$ be the number of modes which become critical as the first Rayleigh number $R_c$ is crossed as given by \eqref{PES}. Ma and Wang \cite{ma2004dynamic, ma2007rayleigh} proved that under some general boundary conditions, the problem has an attractor $\Sigma_R$ which bifurcates from $(0,R_c)$ as $R$ crosses $R_c$. They also proved that the dimension of the bifurcated attractor is $m-1\leq \text{dim}(\Sigma_R) \leq m$. When $m=1$, the structure of $\Sigma_R$ is trivial which is merely a disjoint union of two attracting steady states.

As stated before, in our problem $m$ is either 1 or 2. And the latter case can only happen at a critical length scale $L_c$ where two eigenmodes with consecutive wave numbers become critical.

Numerically, it turns out that the critical Rayleigh numbers for modes with parity 3 or 4 are much greater than those for modes with parity 1 or 2. This can be seen from the Figure~\ref{fig:par34muchbigger}. 

We will assume the following.
\begin{equation} \label{assumptions}
\left\{ 
\begin{aligned}
\text{1. } & \text{$(\beta_1,\phi_1),\, (\beta_2,\phi_2)$ are the first two critical eigenpairs}.\\
\text{2. } & \text{$\phi_1$ has wave number $k$, $\phi_2$ has wave number $k+1$}\\
& \text{ where $k$ is a positive integer.} \\
\text{3. } & \text{One of the eigenmodes $\{\phi_1$, $\phi_2\}$ is of parity class 1, }\\
& \text{while the other is of parity class 2 as given in Table~\ref{tab:parity}}.
\end{aligned}\right.
\end{equation}

Let $y_1$ and $y_2$ be the amplitudes of $\phi_1$ and $\phi_2$ respectively. Then in the proof of the main theorem, we show that the dynamics of the system close to $R=R_c$ and $L = L_c$ is governed by the equations
\begin{equation} \label{transequ}
\begin{aligned}
& \frac{dy_1}{dt} = \beta_1 y_1 + y_1 (a_{11} y_1^2 + a_{13} y_2^2) + o(3),\\
& \frac{dy_2}{dt} = \beta_2 y_2 + y_2 (a_{22} y_1^2 + a_{24} y_2^2) + o(3).\\
\end{aligned}
\end{equation}
Here $\beta_i$ is the eigenvalue corresponding to mode $i$, and $o(3)$ denotes
\[
o(3) = o(  | (y_1,y_2)| ^{3} ) +O( |(y_1,y_2)| ^{3} \max_{i=1,2} |\beta  (R)|) .
\]
Let us define
\begin{equation} \label{D1D2D3}
D_1 = a_{22} \beta_1 - a_{11} \beta_2, \quad D_2 = a_{13} \beta_2 - a_{24} \beta_1, \quad  D_3 = a_{11} a_{24} - a_{13}a_{22}.
\end{equation}
To state our main theorems, we assume the following non-degeneracy conditions
\begin{equation}\label{assumptions2}
a_{11} \neq 0, a_{24}\neq0, D_1 \neq 0, D_2 \neq 0, D_3 \neq 0.
\end{equation}

Finally, let us define the following.
\begin{equation} \label{thmsteadystates}
\begin{aligned}
& \varphi_i = (-1)^i \phi_1, 		&& i=1,2, 		&&& \text{(modes with wavenumber $k$)}\\
& \varphi_i = (-1)^i \phi_2, 		&& i=3,4, 		&&& \text{(modes with wavenumber $k+1$)} \\
& \varphi_i = c_i \phi_1+d_i \phi_2,  && i=5,6,7,8, 	&&& \text{(mixed modes)}
\end{aligned}
\end{equation}
where $c_5 =c_6 = -c_7=-c_8$ and $d_5=-d_6 = d_7 =-d_8$.

\begin{figure}
\centering
\subfigure[$D_1<0$, $D_2<0$, $D_3<0$]{
\includegraphics[scale=.30]{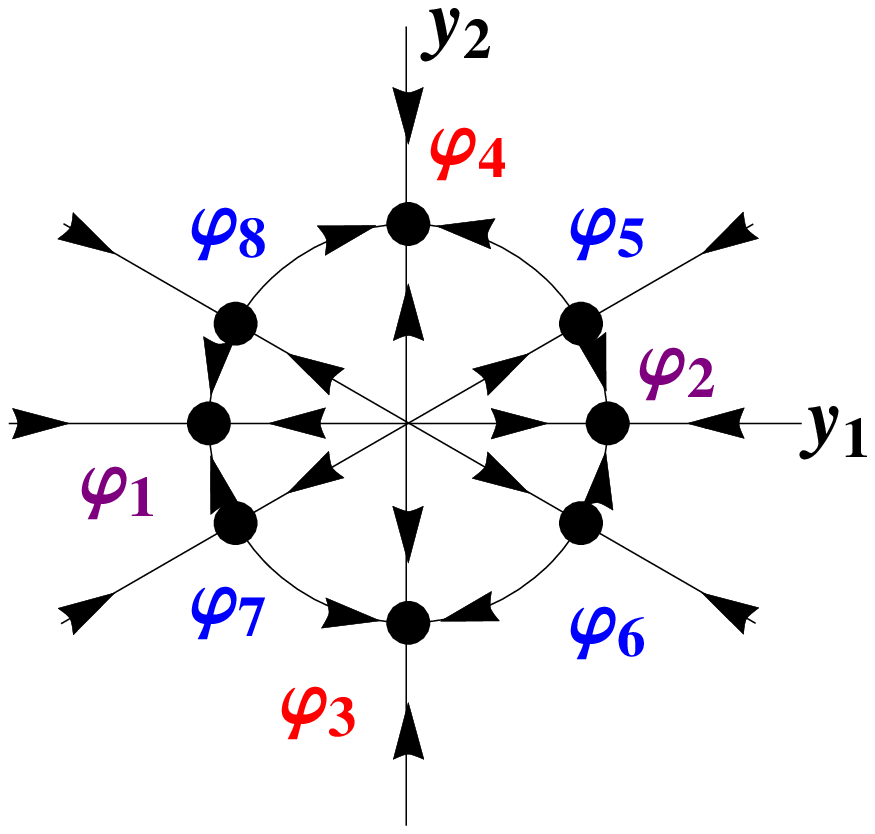}
}
\subfigure[$D_1>0$, $D_2>0$, $D_3>0$]{
\includegraphics[scale=.30]{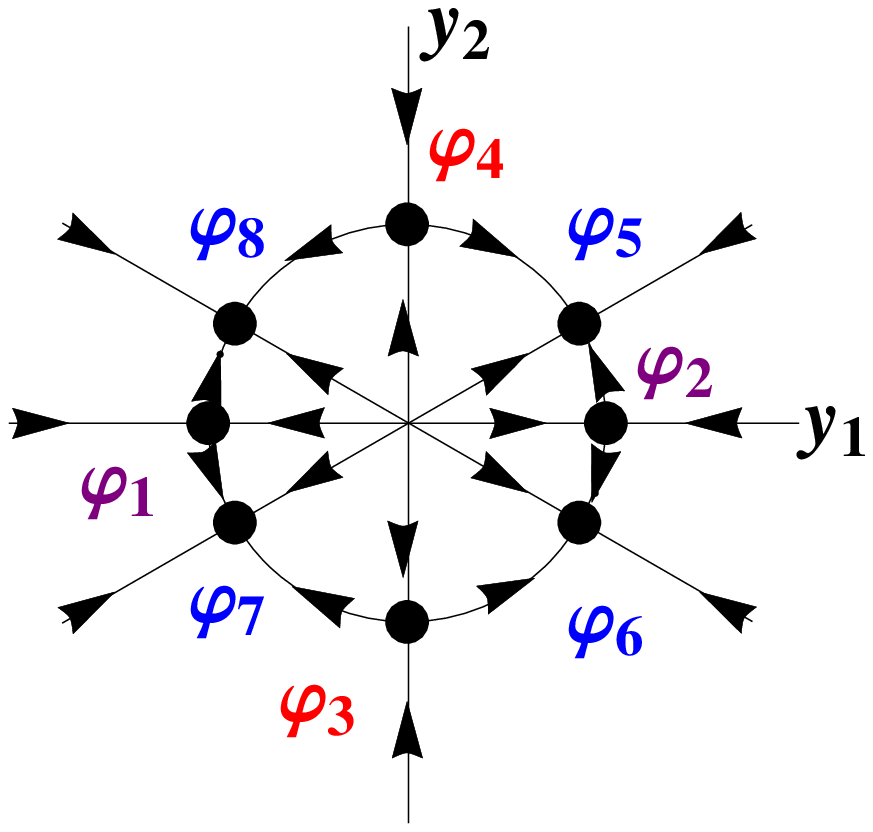}
}
\subfigure[$D_1<0$, $D_2>0$]{
\includegraphics[scale=.36]{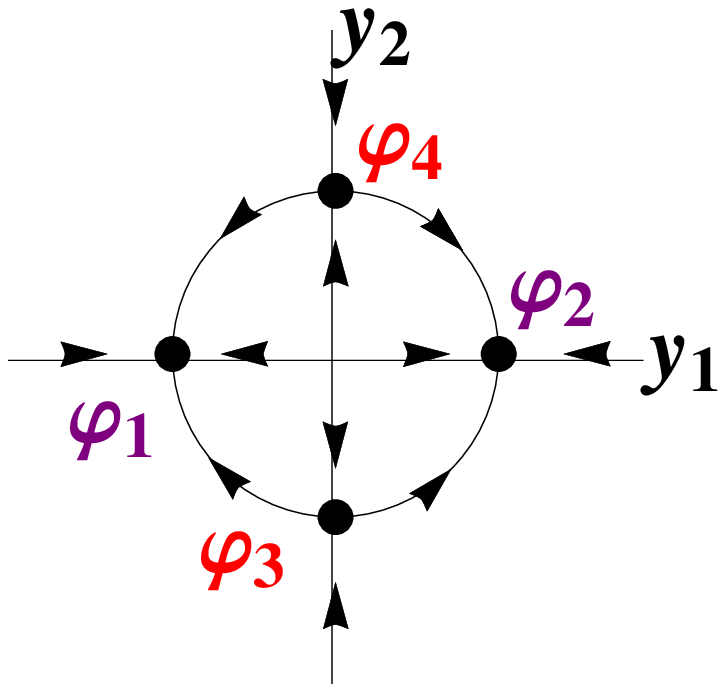}
}
\subfigure[$D_1>0$, $D_2<0$]{
\includegraphics[scale=.36]{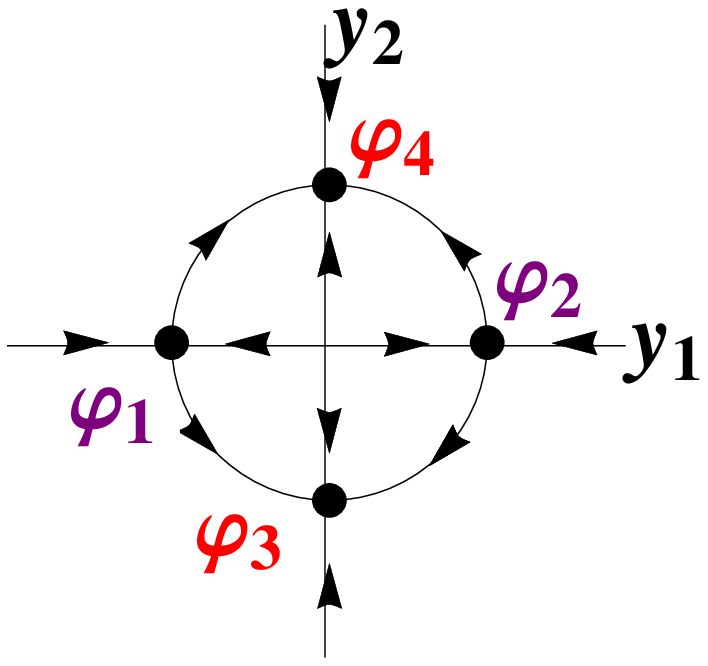}
}
\caption{Transition Scenarios. \label{Transition Scenarios}}
\end{figure}

\begin{theorem}\label{mainthm1}
Under the assumptions \eqref{assumptions} and \eqref{assumptions2}, there is an attractor $\Sigma_R$ bifurcating as $R$ crosses $R_c$ which is homeomorphic to the circle $S^1$ when $L$ is sufficiently close to a critical length scale $L_c$. Moreover, $\Sigma_R$ consists of steady states and their connecting heteroclinic orbits. Let $n(\Sigma_R)$ denote the number of steady states on $\Sigma_R$, $\mathcal{S}$ denote the stable steady states and $\mathcal{U}$ denote the unstable steady states on $\Sigma_R$ upto topological equivalency. Then we have the following characterization of $\Sigma_R$ which is also illustrated in Figure~\ref{Transition Scenarios}.
\begin{itemize}
\item[(i)] If $D_1<0$, $D_2<0$, $D_3<0$, $n(\Sigma_R)=8$, $\mathcal{S} = \{ \varphi_i \mid i=1,2,3,4\}$, $\mathcal{U} = \{ \varphi_i \mid i=5,6,7,8\}$. 
\item[(ii)] If $D_1>0$, $D_2>0$, $D_3>0$, $n(\Sigma_R)=8$, $\mathcal{S} = \{ \varphi_i \mid i=5,6,7,8\}$, $\mathcal{U} = \{ \varphi_i \mid i=1,2,3,4\}$. 
\item[(iii)] If $D_1<0$, $D_2>0$, $n(\Sigma_R)=4$, $\mathcal{S} = \{ \varphi_i \mid i=1,2\}$, $\mathcal{U} = \{ \varphi_i \mid i=3,4\}$. 
\item[(iv)] If $D_1>0$, $D_2<0$, $n(\Sigma_R)=4$, $\mathcal{S} = \{ \varphi_i \mid i=3,4\}$, $\mathcal{U} = \{ \varphi_i \mid i=1,2\}$. \end{itemize}
\end{theorem}

According to Theorem~\ref{mainthm1}, the structure of the attractor depends on $D_1$, $D_2$ and $D_3$ which in turn depends on the coefficients of the reduced equations. By \eqref{D1D2D3}, $D_3$ has a definite sign whereas $D_1$ and $D_2$ vanish at the criticality $\beta_1 =\beta_2 =0$. In the proof of Theorem~\ref{mainthm1}, we analytically prove that the coefficients $a_{11}$, and $a_{24}$ are negative. Our numerical computations indicate that $a_{13}$ is also always negative. We observed that  $a_{22}$ and $D_3$ can be both positive and negative. 

That gives three possible cases depending on the signs of $a_{22}$ and $D_3$. In Figure~\ref{Transition Scenarios2}, we classify these cases in a small neighborhood of $\beta_1 = \beta_2 =0$ in the $\beta_1$--$\beta_2$ plane according to our main theorem  and the following observations.
\begin{itemize}
\item If $a_{22}>0$ then $D_3>0$, $D_1>0$ but  $D_2$ changes sign in the first quadrant.
\item If $a_{22}<0$ and $D_3>0$ then $D_1$ and $D_2$ changes sign in the first quadrant. Moreover the case where both $D_1<0$ and $D_2<0$ is not possible.
\item If $a_{22}<0$ and $D_3<0$ then again $D_1$ and $D_2$ changes sign in the first quadrant. This time the case where both $D_1>0$ and $D_2>0$ is not possible.
\end{itemize}
\begin{figure}
\centering
\subfigure[$a_{22}>0$]{
\includegraphics[scale=.6]{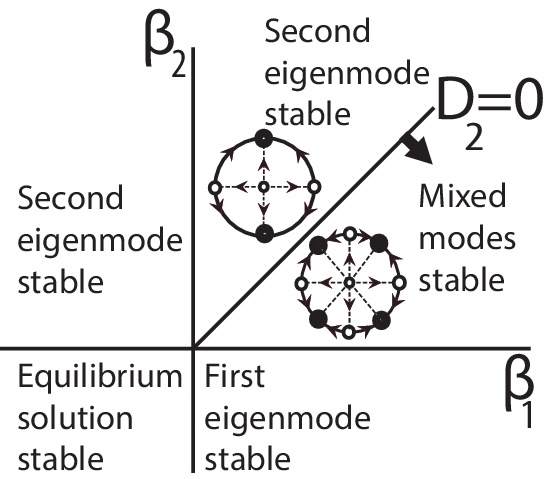}
}
\subfigure[$a_{22}<0, D_3>0$]{
\includegraphics[scale=.6]{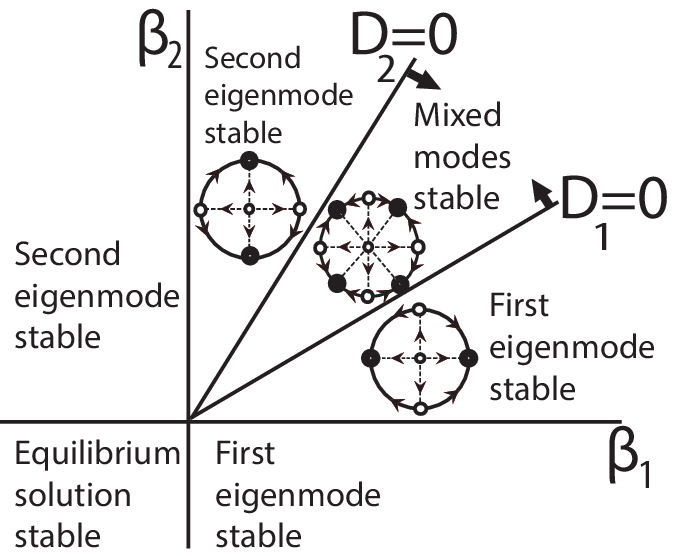}
}
\subfigure[$a_{22}<0,  D_3<0$]{
\includegraphics[scale=.6]{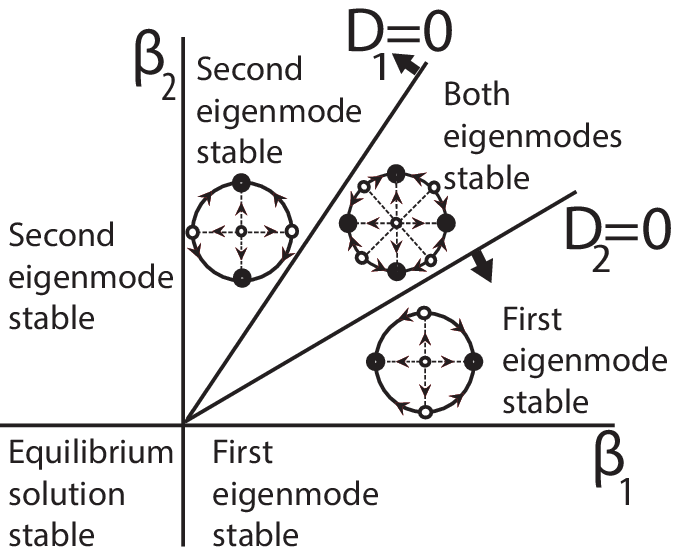}
}
\caption{The transition scenarios in $\beta_1$--$\beta_2$ plane. In the above cases we assume that $a_{13}<0$ which is due to our numerical observations. The arrows on the lines $D_1 =0$, $D_2=0$ indicate in which directions $D_1$ and $D_2$ increase. First and second eigenmodes correspond to the eigenmodes with wavenumber $k$ and $k+1$ respectively. \label{Transition Scenarios2}}
\end{figure}

\section{Proof of the Main Theorem}
We will give the proof in several steps.\\
{\bf STEP 1. The reduced equations.}
When there are two critical modes $\phi_1$, $\phi_2$, the center manifold is a two dimensional manifold embedded in the infinite dimensional space. We denote the center manifold function by:
\[
\Phi = y_1^2 \Phi_1 + y_1 y_2 \Phi_2+ y_2^2 \Phi_3 +o(y^2), \qquad \Phi_i = \left(\Psi_i,  \Theta_i\right)^T.
\]
To study the dynamics on the center manifold, we plug in
\begin{equation} \label{rb2dred2deq1}
\phi = y_1 \phi_1 + y_2 \phi_2 + y_1^2\Phi_1 + y_1 y_2 \Phi_2 +y_2^2 \Phi_3 +o(2),
\end{equation}
into \eqref{functional}, take the inner product with $\phi_1$, $\phi_2$ and use the orthogonality of the eigenvectors, thanks to the self-adjointness of the linear operator. The reduced equations read
\begin{equation} \label{rb2dred2deq3}
\frac{dy_i}{dt} = \beta_i(R) y_i + \frac{1}{(\phi_i,\phi_i)}(G(\phi),\phi_i),  \quad i=1,2
\end{equation}
We normalize the first two eigenfunctions so that
\[
(\phi_1,\phi_1) = (\phi_2,\phi_2) =1.
\]
Now if we expand the nonlinear terms in \eqref{rb2dred2deq3}, we get
\begin{equation} \label{rb2dred2deq4}
\begin{aligned}
& \frac{dy_1}{dt} = \beta_1 y_1 +(a_{11} y_1^3 + a_{12} y_1^2 y_2 + a_{13} y_1 y_2^2 + a_{14} y_2^3) + o(3),\\
& \frac{dy_2}{dt} = \beta_2 y_2 +(a_{21} y_1^3 + a_{22} y_1^2 y_2 + a_{23} y_1 y_2^2 + a_{24} y_2^3) + o(3),\\
\end{aligned}
\end{equation}
where
\begin{equation} \label{a_ij}
\begin{aligned}
& a_{k1} = G_s(\phi_1,\Phi_1,\phi_k), 
&& a_{k2} = G_s(\phi_1,\Phi_2,\phi_k) +G_s(\phi_2,\Phi_1,\phi_k),\\
& a_{k4} = G_s(\phi_2,\Phi_3,\phi_k),
&& a_{k3} = G_s(\phi_1,\Phi_3,\phi_k) +G_s(\phi_2,\Phi_2,\phi_k).\\
\end{aligned}
\end{equation}
{\bf STEP 2. Parities of the center manifold functions.}
To compute the center manifold approximation, we will use the the following formula which was introduced in  Ma--Wang \cite{bifbook}.
\begin{equation} \label{2dcmequ}
\begin{aligned}
& -\mathcal{L}_R \Phi_1 = P_2 G( \phi_1,\phi_1), \\
& -\mathcal{L}_R \Phi_2 = P_2[G( \phi_1,\phi_2)+G(\phi_2,\phi_1)], \\
& -\mathcal{L}_R \Phi_3 =  P_2G(\phi_2,\phi_2).
\end{aligned}
\end{equation}
Here
\[
\begin{aligned}
& P_2 : H \rightarrow E_2, && \mathcal{L}_R = L_R |_{E_2} : E_2 \rightarrow \bar{E_2}, \\
& E_1 = \text{span}\{\phi_1,\phi_2\},
&& E_2 = E_1^{\perp}.
\end{aligned}
\]

Let $X = \{ f\in C(\Omega) \mid f(-x,z) = \pm f(x,z) \text{ and } f(x,-z) = \pm f(x,z) \}$ and let  $s:X\rightarrow \{\pm 1\}^2$ denote the parity function:
\[
s(f) = (s_x(f),s_z(f)),
\]
where
\[
s_x(f) = \pm 1 \quad \text{if } f(-x,z)=\pm f(x,z), \quad s_z(f) = \pm 1 \quad \text{if } f(x,-z)=\pm f(x,z).
\]

Let us define for $\phi_i = (u_i,w_i, \theta_i)$, $i=1,2$ the following.
\begin{equation} \label{g1-g3}
G(\phi_i,\phi_j) = \mat{g_1(\phi_i,\phi_j) \\ g_2(\phi_i,\phi_j) \\ g_3(\phi_i,\phi_j)} = 
\mat{ 	-u_i \frac{\partial u_j}{\partial x} - w_i\frac{\partial u_j}{\partial z} \\
		-u_i \frac{\partial w_j}{\partial x} - w_i\frac{\partial w_j}{\partial z}\\
		-u_i \frac{\partial \theta_j}{\partial x} - w_i\frac{\partial \theta_j}{\partial z}
	}.
\end{equation}

The following lemma can be proved using the basic properties of parities.
\begin{lemma}\label{paritylemma1}
If $\phi_i = (u_i,v_i,\theta_i) \in X^3\cap H_1$, $i=1,2$ then for $i,j,k =1,2$, 
\begin{itemize}
\item[1)] $-s(g_1(\phi_i,\phi_j))$ $= s(g_2(\phi_i,\phi_j))$ $= s(g_3(\phi_i,\phi_j))$ $= (s_x(w_i w_j), -s_z(w_i w_j))$.
\item[2)] $s(g_k(\phi_i,\phi_j)) = s(g_k(\phi_j,\phi_i))$.
\end{itemize}
\end{lemma}

Hereafter without loss of generality we will assume
\begin{equation} \label{parity assumption}
\text{ $\phi_1$ is of parity class 1 and $\phi_2$ is of parity class 2,}
\end{equation}
which are as given in Table~\ref{tab:parity}.

Using the Lemma~\ref{paritylemma1}, we can prove
\begin{lemma}\label{paritylemma2}
Under the assumption \eqref{parity assumption}, 
\[
s(g_2(\phi_1,\phi_1)) =  s(g_2(\phi_2,\phi_2)) =  (1,-1), \qquad s(g_2(\phi_1,\phi_2)) = (-1,-1).
\]
\end{lemma}

\begin{lemma}\label{P2Glemma} 
Under the assumption \eqref{parity assumption},
$P_2G(\phi_i,\phi_j) = G(\phi_i,\phi_j)$ for $i,j=1,2$.
\end{lemma}
\begin{proof}
Note that $P_2G(\phi_i,\phi_j) = G(\phi_i,\phi_j)$ if $(G(\phi_i,\phi_j),\phi_k) = 0$ for $i,j,k = 1,2$.
Now
\begin{equation} \label{P2Glemmaequ}
(G(\phi_i,\phi_j),\phi_k) = \int_{\Omega} \left(g_1(\phi_i,\phi_j) u_k + g_2(\phi_i,\phi_j) w_k + g_3(\phi_i,\phi_j) \theta_k \right) dx dz
\end{equation}

By Lemma~\ref{paritylemma1} and Lemma~\ref{paritylemma2}, $g_1(\phi_i,\phi_j$) is even in the $z$-direction while $g_2(\phi_i,\phi_j)$ and $g_3(\phi_i,\phi_j)$ are odd in the $z$-direction.  Since $u_k$ is odd and $w_k$ and $\theta_k$ are even in the $z$-direction, the integral in \eqref{P2Glemmaequ} must vanish over $\Omega$. 
\end{proof}

Thus by the Lemma~\ref{P2Glemma} and the equation \eqref{2dcmequ},  $\Phi_i=(U_i,W_i,\Theta_i)$, $(i=1,2,3)$ are solutions of
\begin{equation} \label{2dcmequ2}
\begin{aligned}
& -\mathcal{L}_R \Phi_1 = G( \phi_1,\phi_1), \\
& -\mathcal{L}_R \Phi_2 = G( \phi_1,\phi_2)+G(\phi_2,\phi_1), \\
& -\mathcal{L}_R \Phi_3 = G(\phi_2,\phi_2).
\end{aligned}
\end{equation}
Using the streamfunction $\Psi_z = U$, $\Psi_x = -W$, one can eliminate the pressure from these equations to obtain
\begin{equation}\label{2dcmequ4}
\begin{aligned} 
& \Pr \Delta^2\Psi-\sqrt{R} \sqrt{\Pr} \frac{\partial \Theta}{\partial x} = h_1 := -\frac{\partial g_1}{\partial z}+\frac{\partial g_2}{\partial x},\\
-& \sqrt{R} \sqrt{\Pr} \frac{\partial \Psi}{\partial x}+\Delta\theta = h_2 := - g_3,\\
& \Psi=\frac{\partial \Psi}{\partial n}=\Theta=0 \text{ on }\partial \Omega.
\end{aligned}
\end{equation}

\begin{lemma}
Under the assumption \eqref{parity assumption},  the center manifold functions have the parity as given in Table~\ref{cmparity}.
\end{lemma}
\begin{proof}
We can eliminate $\Theta$ from the first equation of \eqref{2dcmequ4} to obtain
\begin{equation} \label{}
\begin{aligned}
& \Pr \Delta^3\Psi - R \Pr  \frac{\partial^2 \Psi}{\partial x^2}=  \Delta h_1+ \sqrt{R} \sqrt{\Pr}\frac{\partial h_2}{\partial x},\\
& \Delta \Theta = h_2 +\sqrt{R} \sqrt{\Pr}\frac{\partial \Psi}{\partial x}.
\end{aligned}
\end{equation}
Now using Lemma~\ref{paritylemma1} and Lemma~\ref{paritylemma2}, we see that    $s(\Psi) = (-s_x(g_2),s_z(g_2))$ and $s(\Theta)  = s(g_2)$.
\end{proof}

\begin{table}
\begin{center}
\begin{tabular}{|c|c|c|c|c|c|c|c|c|c|}
\hline
\multicolumn{2}{|c|}{$\phi_1$} & \multicolumn{2}{|c|}{$\phi_2$} & \multicolumn{2}{|c|}{$\Phi_1$}  & \multicolumn{2}{|c|}{$\Phi_2$}  &\multicolumn{2}{|c|}{$\Phi_3$} \\\hline
$\psi_1$ & $\theta_1$ & $\psi_2$ & $\theta_2$ & $\Psi_1$ & $\Theta_1$ & $\Psi_2$ & $\Theta_2$ & $\Psi_3$ & $\Theta_3$ \\ \hline
(e,e)&(o,e)&(o,e)&(e,e)&(o,o)&(e,o)&(e,o)&(o,o)&(o,o)&(e,o)\\\hline
\end{tabular}
\end{center}
\caption{\label{cmparity} Parities of the first two critical modes and the center manifold functions.}
\end{table}

Using Table~\ref{cmparity} we find that the integrands in $a_{12}$, $a_{14}$, $a_{21}$, $a_{23}$ are all odd functions of z and hence we have the following.
\begin{lemma}\label{somecozero}
Under the assumption \eqref{parity assumption}, in \eqref{a_ij} we have 
\[
a_{12} = a_{14} = a_{21} = a_{23} = 0.
\]
\end{lemma}
As a result of Lemma~\ref{somecozero}, we obtain the reduced equations \eqref{transequ}.

{\bf STEP 3. The attractor bifurcation.}
Now, we will prove that the bifurcated attractor is homeomorphic to $S^1$. For this we will need the following result.
\begin{theorem}[Ma--Wang \cite{bifbook}]\label{S1biftheorem}
Let $v$ be a two dimensional $C^r\, (r\geq1)$ vector field given by
\begin{equation} \label{finS1}
v_{\lambda} (x) = \beta(\lambda) \, x - h(x,\lambda),
\end{equation}
for $x\in \mathbb{R}^2$. Here $\beta(\lambda)$ is a continuous function of $\lambda$ satisfying $\beta(\lambda)\overset{>}{\underset{<}{=}}0$ for $ \lambda \overset{>}{\underset{<}{=}} \lambda_0 $ and
\[
h(x,\lambda) = h_k(x,\lambda) +o(|x|^k), \quad C_1 |x|^{k+1}\leq (h_k(x,\lambda),x),
\]
for some odd integer $k\geq3$ where $h_k(\cdot, \lambda)$  is a k-multilinear field, and $C_1>0$ is some constant. Then the system
\[
dx/dt = v_{\lambda}(x), \quad x\in \mathbb{R},
\]
bifurcates from $(x,\lambda) = (0,\lambda_0)$ to an attractor $\Sigma_{\lambda}$, which is homeomorphic to $S^1$, for $\lambda_0<\lambda<\lambda_0 +\epsilon$, for some $\epsilon>0$. Moreover, either (i) $\Sigma_{\lambda}$ is a periodic orbit, or (ii) $\Sigma_{\lambda}$ consists of an infinite number of singular points, or, (iii) $\Sigma_{\lambda}$ contains at most $2(k+1)$ singular points, consisting of $2N$ saddle points, $2N$ stable node points and $n(\leq 2(k+1) -4N)$ singular points with index zero.
\end{theorem}
Now let
\[
h(y_1,y_2) = \left\{ y_1 (a_{11} y_1^2 + a_{13} y_2^2), \, y_2 (a_{22} y_1^2 + a_{24} y_2^2)\right\}^T.
\]

\begin{lemma} \label{S1lemma}
Assume that $\Phi_i \neq 0$ for $i=1,2,3$. Then for any $y=(y_1,y_2)$,
\begin{equation} \label{(hy,y)}
(h(y), y) = a_{11} y_1^4 +(a_{13} + a_{22}) y_1^2 y_2^2 + a_{24} y_2^4 \leq C |y|^4,
\end{equation}
where $C<0$.
\end{lemma}
\begin{proof}
\begin{equation} \label{a11a24}
\begin{split}
a_{11} &=  G_s(\phi_1,\Phi_1,\phi_1)  =  G(\phi_1,\Phi_1,\phi_1) + G(\Phi_1, \phi_1, \phi_1) \\
& = G(\phi_1,\Phi_1,\phi_1)= -G(\phi_1,\phi_1,\Phi_1) =  (\mathcal{L}_R \Phi_1, \Phi_1).
\end{split}
\end{equation}
Here we used \eqref{2dcmequ2} and the following properties of Navier-Stokes nonlinearity
\begin{equation} \label{navstoknon}
(i)\, G(\phi,\tilde{\phi},\phi^{\ast})= G(\phi,\phi^{\ast},\tilde{\phi}), \qquad 
(ii)\, G(\phi,\tilde{\phi},\tilde{\phi})= 0,
\end{equation}
and $-\mathcal{L}_R \Phi_1 = G(\phi_1,\phi_1)$ which is due \eqref{2dcmequ2}. 

If we write
\[
\Phi_j = \sum_{k=3}^{\infty} c_{j,k} \phi_k, \quad j =1,2,3,
\]
then for $j=1,2,3$, 
\begin{equation} \label{LRPjPj}
(\mathcal{L}_R \Phi_j, \Phi_j) = \sum_{k=3}^{\infty} c_{j,k}^2 \beta_k || \phi_k ||^2 <0.
\end{equation}
since $\beta_k<0$ for $k\geq3$ and by assumption there exists $k\geq 3$ such that $c_{1,k} \neq0$.  In particular, $a_{11}<0$. As in \eqref{a11a24}, we can show that
\[
a_{24} = G_s(\phi_2,\Phi_3,\phi_2) = (\mathcal{L}_R \Phi_3, \Phi_3)<0.
\]

Now if $a_{13}+a_{22}<0$ then it is easy to prove \eqref{(hy,y)}. Assume otherwise. Using \eqref{navstoknon} and \eqref{2dcmequ2}, we can write 
\begin{equation} \label{a13}
\begin{split} a_{13} & =  G_s(\phi_1,\Phi_3,\phi_1) + G_s (\phi_2,\Phi_2,\phi_1) \\
& = G(\phi_1,\Phi_3,\phi_1) + G(\Phi_3,\phi_1,\phi_1) + G_s (\phi_2,\Phi_2,\phi_1) \\
& = - (G(\phi_1),\Phi_3) + G_s (\phi_2,\Phi_2,\phi_1) \\
& = (\mathcal{L}_R \Phi_1,\Phi_3) + G_s (\phi_2,\Phi_2,\phi_1). \\
\end{split}
\end{equation}
A similar computation shows
\begin{equation} \label{a22}
a_{22}  =  (\mathcal{L}_R \Phi_3,\Phi_1) + G_s (\phi_1,\Phi_2,\phi_2).
\end{equation}
Let us define
\begin{equation} \label{alpha}
\alpha = G_s (\phi_1,\Phi_2,\phi_2)+ G_s (\phi_2,\Phi_2,\phi_1).
\end{equation}
By \eqref{navstoknon} and \eqref{2dcmequ2},
\begin{equation} \label{alphaineq}
\alpha  = -(G(\phi_1,\phi_2)+G(\phi_2,\phi_1),\Phi_2) = (\mathcal{L}_R \Phi_2, \Phi_2)
\end{equation}
Note that $\alpha<0$ by \eqref{LRPjPj}. Using Cauchy-Schwarz inequality and the orthogonality of the eigenfunctions,
\begin{equation} \label{LRP1P3}
\begin{split}
(\mathcal{L}_R \Phi_1,\Phi_3)& = \sum_{k=3}^{\infty} \beta_k c_{1,k} c_{3,k} ||\phi_k||^2 \\
 & \leq  \left(\sum_{k=3}^{\infty} -\beta_k c_{1,k}^2  ||\phi_k||^2 \right)^{1/2} \left(\sum_{k=3}^{\infty} -\beta_k c_{3,k}^2  ||\phi_k||^2 \right)^{1/2} \\
 & = \sqrt{a_{11}a_{24}}.
 \end{split}
\end{equation}
Since,  $(\mathcal{L}_R \Phi_3,\Phi_1)=(\mathcal{L}_R \Phi_1,\Phi_3)$, we have  by \eqref{a13}--\eqref{LRP1P3},
\[
 a_{13} + a_{22} < 2 \sqrt{a_{11}a_{24}} + \alpha,
\]
where $\alpha<0$ is given by \eqref{alpha}. Thus, there exists $0<\epsilon_1<-a_{11}$, $0<\epsilon_2<-a_{24}$ such that
\[
 a_{13} + a_{22} < 2 \sqrt{a_{11}a_{24}} + \alpha < 2 \sqrt{(a_{11}+\epsilon_1) (a_{24}+\epsilon_2)}.
\]
Since $2 ab <a^2 +b^2$, we have,
\[
2 \sqrt{(a_{11}+\epsilon_1) (a_{24}+\epsilon_2)}y_1^2 y_2^2\leq -(a_{11}+\epsilon_1) y_1^4 - (a_{24}+\epsilon_2) y_2^4.
\]
Now, let $C =\max \{ -\epsilon_1,-\epsilon_2\}$. Then $C<0$ and we have
\[
(h(y),y) \leq a_{11} y_1^4 +(a_{13} +a_{22}) y_1^2y_2^2 +a_{24} y_2^4 \leq C(x^2 + y^2)^2.
\]
That finishes the proof.
\end{proof}

Thus  by Theorem~\ref{S1biftheorem} and Lemma~\ref{S1lemma}, $\Sigma_{R}$ is homeomorphic to $S^1$. Now we will describe the details of its structure by determining the bifurcated steady states and their stabilities.

{\bf STEP 4. The steady states and their stabilities.}
The possible equilibrium solutions of the truncated equations of \eqref{transequ} are as follows.
\begin{equation} \label{steadystates}
R_1 = (\sqrt{\frac{\beta_1}{-a_{11}}} , 0), \, R_2 = (0, \sqrt{\frac{\beta_2}{-a_{24}}}), \, M = (\sqrt{\frac{D_2}{D_3}},\sqrt{\frac{D_1}{D_3}}),
\end{equation}
where $D_1$, $D_2$ and $D_3$ are given by \eqref{D1D2D3}.

Due to the invariance of the equations \eqref{transequ} with respect to $(x,y) \rightarrow (-x,y)$ and $(x,y) \rightarrow (x,-y)$, we only consider the positive solutions when writing \eqref{steadystates}.

The eigenvalues $\lambda_1$, $\lambda_2$ of the truncated vector field at the steady states $R_1$, $R_2$ are
\begin{equation*}
\begin{aligned}
& \lambda_1^{R_1} =  -2 \beta_1, \, \lambda_2^{R_1}  = -D_1/a_{11}, \, \lambda_1^{R_2} =  -2 \beta_2, \, \lambda_2^{R_2}  = -D_2/a_{24}.
\end{aligned}
\end{equation*}

Note that $R_i$ is always bifurcated for $\beta_i>0$, $i=1,2$. Moreover $R_i$ is a stable steady state for $\beta_i>0$ if $D_i<0$ for $i=1,2$.
The trace $Tr$ and the determinant $Det$ of the Jacobian matrix of the truncated vector field at the mixed states $M$ are
\begin{equation} \label{stabilityassignments}
Tr = \frac{2}{D_3} (a_{24} D_1 + a_{11} D_2), \qquad Det = \frac{4}{D_3} D_1 D_2. 
\end{equation}
Notice that the steady states $M$ are bifurcated only when  $D_1$, $D_2$, $D_3$ have the same sign. Since $a_{11}$ and $a_{24}$ are both negative as shown in Lemma 5.7, according to trace-determinant plane analysis, they are saddles if $D_1<0$, $D_2<0$, $D_3<0$ and are stable if $D_1>0$, $D_2>0$, $D_3>0$.

Finally, only the four cases stated in our main theorem can occur. To see this, note that according to the Theorem~\ref{S1biftheorem} and \eqref{steadystates}--\eqref{stabilityassignments}, the case $D_1<0$, $D_2<0$, $D_3>0$ is not possible since that would lead to only 4 steady states on the attractor which are all stable. Similarly the case $D_1>0$, $D_2>0$, $D_3<0$ is not possible either which would lead to 4 steady states which are all unstable.

\section{Numerical approximation of the coefficients of the reduced equations}
To compute the coefficients of the reduced equations \eqref{transequ}, we fix $L$, $\Pr$ and $R$  to compute all the eigenvalues $\beta_i^N$ and the corresponding eigenvectors of \eqref{rb2d4}.

{\bf Numerical computation of the center manifold functions.} Now we will numerically approximate $\Phi_1$,  $\Phi_2$ and $\Phi_3$ which are the solutions of the equations \eqref{2dcmequ}. We will illustrate the method to approximate $\Phi_1$ since $\Phi_2$, $\Phi_3$ can be approximated similarly. 
To determine $\Phi_1$, we have to find its stream function $\Psi$ and its temperature function $\Theta$ which are determined by equations \eqref{2dcmequ4}.

Since we do not have $h_1$ and $h_2$ in \eqref{2dcmequ4} exactly, we approximate them by $h_1^N$, $h_2^N$ as below
\begin{equation} \label{h1N-h2N}
\begin{aligned}
& h_1^N = -\frac{\partial g_1^N}{\partial z}+\frac{\partial g_2^N}{\partial x}, && h_2^N = \psi^N_{1,z}\theta^N_{1,x} - \psi^N_{1,x} \theta^N_{1,z},\\
& g_1^N = -\psi^N_{1,z}\psi^N_{1,xz} + \psi^N_{1,x} \psi^N_{1,zz}, && g_2^N = -\psi^N_{1,z}\psi^N_{1,xx} + \psi^N_{1,x} \psi^N_{1,xz}.
\end{aligned}
\end{equation}
Here ($\psi_1^N$, $\theta_1^N$) is the first critical eigenfunction of the discrete problem \eqref{rb2d4}.
\begin{equation} \label{2dcmequ4.5}
 \{\psi_1^N, \theta_1^N\}  = \sum_{m=0}^{N_x-1} \sum_{n=0}^{N_z-1}\left\{ \tilde{\psi}^N_{1,mn} e_m(x) e_n(z), \,\tilde{\theta}^N_{1,mn} f_m(x) f_n(z) \right\}.
\end{equation}

{\bf The Legedre-Galerkin approximation of the problem \eqref{2dcmequ4}.}
As in the linear eigenvalue problem, we discretize the equations \eqref{2dcmequ4} using the generalized Jacobi polynomials \eqref{J11}--\eqref{J22}.
\begin{equation} \label{rbcmeq2}
\{\Psi^N, \Theta^N \} = \sum_{m=0}^{N_x-1} \sum_{n=0}^{N_z-1} \left \{\tilde{\Psi}^N_{mn} e_m(x) e_n(z), \,  \tilde{\Theta}^N_{mn} f_m(x) f_n(z) \right\}.
\end{equation}

We plug in $\Psi^N$, $\Theta^N$, $h_1^N$, $h_2^N$ for $\Psi$, $\Theta$, $h_1$, $h_2$ in \eqref{2dcmequ4} and multiply the resulting equations by Jacobi polynomials $e_j(x)e_k(z)$, $f_j(x)f_k(z)$ and integrate over $-1\leq x\leq1$, $-1\leq z\leq1$ to reduce \eqref{2dcmequ4} to the following finite dimensional linear equation
\begin{equation} \label{rbcmeq3}
(B^N - \sqrt{R} \, C^N)\bar{x} = \bar{b}.
\end{equation}
Here $B^N$ and $C^N$ are given by \eqref{rb2d5} and
\begin{equation} \label{bbar}
\bar{x} =  \mat{\text{vec}(\tilde{\Psi}^N) &  \text{vec}(\tilde{\Theta}^N) }^T_{N^2\times1}, \quad
\bar{b} = \mat{ \text{vec}(B_1)  &  \text{vec}(B_2)}^T.
\end{equation}
For $0\leq j \leq N_x-1$, $0\leq k \leq N_z -1$,
\begin{equation} \label{B1-B2}
\begin{aligned}
& (B_1)_{jk} = \int_{-1}^1\int_{-1}^1 h_1^N(x,z) e_j(x)e_k(z) dx dz,\\
& (B_2)_{jk} = \int_{-1}^1\int_{-1}^1\ h_2^N(x,z) f_j(x)f_k(z)dx dz.
\end{aligned}
\end{equation}
Now $e_j$ is a polynomial of degree $j+4$ and by \eqref{h1N-h2N} and \eqref{2dcmequ4.5}, $h_i^N$ is a polynomial of degree at most $(2N_x +6,2N_z+6)$. Thus the above integrands are of degree at most $(3 N_x+9, 3 N_z+9)$. Since the  Legendre-Gauss-Lobatto quadrature with $N+1$ quadrature points is exact for polynomials of degree less or equal than $2N-1$,  the integrals in \eqref{B1-B2} can be replaced by the following discrete inner products.
\begin{equation} \label{B1-B2num}
\begin{aligned}
& (B_1)_{jk} = \sum_{m=0}^{\frac{3}{2}N_x +5} \sum_{n=0}^{\frac{3}{2}N_z +5} h_1^N(x_m,z_n) e_j(x_m) e_k(z_n) \omega^x_m \omega^z_n,\\
& (B_2)_{jk} = \sum_{m=0}^{\frac{3}{2}N_x +5} \sum_{n=0}^{\frac{3}{2}N_z +5} h_2^N(x_m,z_n) f_j(x_m) f_k(z_n) \omega^x_m \omega^z_n.
\end{aligned}
\end{equation}
Here  $\{x_j,w^x_j\}_{j=0}^{\frac{3}{2}N_x +5}$ and $\{z_j,w^z_j\}^{\frac{3}{2}N_z +5}_{j=0}$ are the  Legendre-Gauss-Lobatto points and weights in the x-direction and the z-direction.


{\bf Solution of \eqref{rbcmeq3}.}
The solution $\bar{x}$ of \eqref{rbcmeq3} can be obtained by inverting the matrix $(B^N - \sqrt{R} \, C^N)$. But this matrix has a large condition number. Thus we show a method to obtain the solution inverting the matrix $D^N$ given by \eqref{rb2d5} which has a much smaller condition number. For example, for $N_x =10$, $N_z =8$, the condition number of $(B^N - \sqrt{R} \, C^N)$ is $O(10^{16})$ while the condition number of $D^N$ is $O(10^8)$. 

Since $\Phi_1 \in E_2=\text{span}\{\phi_1,\phi_2\}^{\perp}$, we look for a solution of \eqref{rbcmeq3} in the form
\begin{equation} \label{rb2dsolexp}
\bar{x} = \sum_{i =3}^N x_i \bar{x}_i,
\end{equation}
where $\bar{x}_i$ are the eigenvectors of 
\begin{equation} \label{nonsoleig}
B^N\bar{x}_i - \sqrt{R} C^N \bar{x}_i = \beta_i(R) D^N\bar{x}_i.
\end{equation}
If we multiply \eqref{rbcmeq3} by $(D^N)^{-1}$ and use \eqref{nonsoleig}, the left hand side of \eqref{rbcmeq3} becomes
\begin{equation} \label{rbcmeq3lhs}
 \sum_{i =3}^N x_i \beta_i(R) \bar{x}_i =(D^N)^{-1}(B^N-\sqrt{R} C^N) \bar{x} =(D^N)^{-1} \bar{b}:=\bar{f}.
\end{equation}
We determine $\bar{f}$ from $D^N\bar{f} = \bar{b}$ using Gaussian elimination. Once again using Gaussian elimination, we can find the coefficients $f_i$ in the expansion
\begin{equation} \label{rb2dbarb}
\bar{f} = \sum_{i = 1}^N f_i \bar{x}_i.
\end{equation}
In \eqref{rb2dbarb}, we see that $f_1 = f_2 = 0$ is necessary for the existence of a solution of \eqref{rbcmeq3}. From \eqref{rbcmeq3lhs} and \eqref{rb2dbarb}, one finds $x_i = f_i/\beta_i(R)$, $i =3,4,\dots, N$.
Thus  the Jacobi expansion coefficients in \eqref{rbcmeq2} of the center manifold are given by
\[
\text{vec}(\tilde{\Psi}^N) =\sum_{i=3}^N \frac{f_i}{\beta_i(R)} \text{vec}(\tilde{\psi}^N_i),  \qquad \text{vec}(\tilde{\Theta}^N) =\sum_{i=3}^N \frac{f_i}{\beta_i(R)} \text{vec}(\tilde{\theta}^N_i).
\]

{\bf Numerical computation of $a_{ij}$ in \eqref{a_ij}.}
We approximate $a_{11}$ by
\begin{equation} \label{approxa_11}
a_{11}^N =  G_s(\phi_1^N,\Phi_1^N,\phi_1^N).
\end{equation}
The integrands in $G_s(\phi_1^N,\Phi_1^N,\phi_1^N)$ are polynomials of degree at most $(3N_x+9,3N_z+9)$. Thus to replace the integrals in \eqref{approxa_11}, one needs again $(\frac{3}{2}N_x+5,\frac{3}{2}N_z+5)$ quadrature points and nodes in the numerical inner product. The other coefficients $a_{ij}$ in \eqref{a_ij} are approximated similarly.
\begin{remark}
We observed that increasing $N_x$ and $N_z$ above $N_x = 10+ 2k$ and $N_z = 8$ only changes $a_{ij}^N$ in the seventh digit when the first critical mode which has $k$ rolls and the second critical mode has $k+1$ rolls in their stream functions. \end{remark}
\section{Numerical Results and Discussion}
We computed coefficients of the reduced equations for various Pr values ranging from $0.1$ to $10^3$ at the first three critical length scales and at the critical Rayleigh numbers which are given in Table~\ref{asymp}. 
 
As proved in Theorem~\ref{mainthm1}, the coefficients $a_{11}$ and $a_{24}$ are always negative. In our numerical calculations, we encountered that $a_{13}$ is also always negative. But the sign of $a_{22}$ and the sign of $D_3$ depends on $L$ and Pr and are given in Table~\ref{coefficients}.
 
For the first critical length scale $L_c = 1.5702$, we found that $a_{22}$ and $D_3$ changes sign from positive to negative between $0.04<\Pr<0.05$ and $0.14<\Pr<0.15$ respectively. Thus the transition is as described in Figure~\ref{Transition Scenarios2}(a) for $\Pr<0.04$, as in Figure~\ref{Transition Scenarios2}(b) for $0.05 < \Pr<0.14$ and as in Figure~\ref{Transition Scenarios2}(c) for $\Pr>0.15$. Thus the mixed modes can be stable when $\Pr<0.14$ but only the pure modes are stable points of the attractor when $\Pr>0.15$.

For the second critical length scale $L_c = 2.6611$, we always observed that $a_{22}<0$. However, $D_3$ changes sign between $0.05<\Pr<0.06$. Thus the transition is as described in Figure~\ref{Transition Scenarios2}(b) for $\Pr<0.05$ and as described in Figure~\ref{Transition Scenarios2}(c) for $\Pr>0.06$. In particular,  the mixed modes can be stable when $\Pr<0.05$ but only the pure modes are stable steady states when $\Pr>0.06$.

For higher critical length scales (third and beyond), we found that $a_{22}<0$ and $D_3<0$ for the Prandtl numbers we considered. Thus the transition is as described in Figure~\ref{Transition Scenarios2}(c). For this length scale, either the critical Prandtl number that was observed for the first two critical length scales is now very close to zero or it does not exist at all.

The above analysis depends on the coefficients $a_{ij}$ of the reduced equations and predicts the transitions when both eigenvalues $\beta_1$, $\beta_2$ are close to zero. Now we present an analysis depending on the direct computation of the numbers $D_1$, $D_2$ (both of which vanish when $\beta_1 = \beta_2 =0$)  and $D_3$.  We computed $D_1$, $D_2$ and $D_3$ for $L$ and $R$ values around (but not necessarily very close to) the criticality $(L,R)=(L_c,R_c)$ for the first three critical length scales and for Prandtl numbers $\Pr=0.1,\, 0.71,\, 7,\, 130$. The results are shown in Figure~\ref{Pr-L-R}. Although we might have omitted the smallness assumptions of $|L-L_c|$ and $|R-R_c|$ where our main theorem is valid, these figures help us predict the transitions in the $L-R$ plane. The results we obtain are as follows.

For $\Pr = 0.71$, $\Pr= 7$, $\Pr=130$, transitions are qualitatively same in the $L$--$R$ plane. For $L>L_c$, the basic motionless state loses its stability to the eigenmode with wavenumber $k+1$ as the Rayleigh number crosses the first critical Rayleigh number and further increase of the Rayleigh number does not alter the stability of this steady state. This is in contrast to the situation $L<L_c$ where there is a transition of stabilities as the Rayleigh number is increased. Namely, as the Rayleigh number crosses the first critical Rayleigh number, the eigenmode with wavenumber $k$ becomes stable. As the Rayleigh number is further increased, both eigenmodes coexist as stable steady states and the initial conditions determine which one of these steady states will be realized. Finally as the Rayleigh number is further increased, the eigenmode with wavenumber $k+1$ becomes stable. 

The transition at $\Pr =0.1$ is essentially different than for those at $\Pr =0.71,\, 7,\, 130$. In particular, for the first critical length scale $L_c = 1.5702$, for $L<L_c$, subsequently mode with wavenumber $k$, mixed modes and finally mode with wavenumber $k+1$ will be realized as the  Rayleigh number is increased while for $L>L_c$,  $k+1$ mode is the only stable steady state.

\begin{figure}
\centering
\includegraphics[scale=.25]{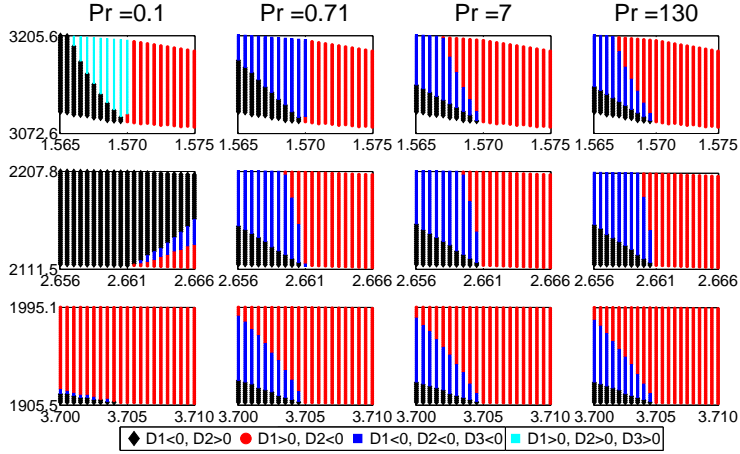}
\caption{\label{Pr-L-R} The signs of $D_1$, $D_2$ and $D_3$ in the $L$--$R$ plane. In each subfigure, the x and y axes denote the length scale L and the Rayleigh number R respectively. For each column, the Prandtl number is given above.}
\end{figure}

\begin{table}
\begin{tabular}{|l|l|l|l|l|l|l|}
\hline
& \multicolumn{2}{|c|}{$L_c = 1.5702$}  & \multicolumn{2}{|c|}{$L_c = 2.6611$} & \multicolumn{2}{|c|}{$L_c = 3.7048$} \\ \hline
\textbf{Pr}&\textbf{a$_{22}\times 10^2$}&\textbf{D$_3\times 10^5$}&\textbf{a$_{22}\times 10^2$}&\textbf{D$_3\times 10^5$}&\textbf{a$_{22}\times 10^2$}&\textbf{D$_3\times 10^5$}\\\hline
0.01&7.6492&2559.8276&-23.7252&140.2894&-20.5185&-15.6696\\\hline
0.04&0.1313&93.3662&-6.0784&3.6690&-5.3093&-1.4416\\\hline
0.05&-0.4546&50.5609&-5.0374&1.0455&-4.4332&-1.1525\\\hline
0.06&-0.8740&29.9252&-4.3902&-0.2713&-3.8966&-1.0122\\\hline
0.14&-2.3697&0.6458&-3.0486&-2.3788&-2.8749&-0.9899\\\hline
0.15&-2.4662&-0.0707&-3.0218&-2.4414&-2.8641&-1.0182\\\hline
0.71&-3.2469&-9.5827&-2.8940&-2.5437&-2.8848&-1.4316\\\hline
7&-0.8551&-0.8348&-0.7484&-0.1610&-0.7444&-0.0987\\\hline
100&-0.0687&-0.0055&-0.0599&-0.0010&-0.0595&-0.0006\\\hline
1000&-0.0069&-0.0001&-0.0060&-0.00001&-0.0060&-0.000006\\\hline
\end{tabular}
\caption{\label{coefficients}The coefficients of the reduced equations for various Pr values at the first three critical length scales.}
\end{table}

\bibliographystyle{amsalpha}
\providecommand{\bysame}{\leavevmode\hbox to3em{\hrulefill}\thinspace}
\providecommand{\MR}{\relax\ifhmode\unskip\space\fi MR }
\providecommand{\MRhref}[2]{%
  \href{http://www.ams.org/mathscinet-getitem?mr=#1}{#2}
}
\providecommand{\href}[2]{#2}

\end{document}